\documentclass[journal,comsoc]{IEEEtran}
\usepackage{graphicx}
\usepackage{amssymb}
\usepackage{amsmath}
\usepackage{cite}
\usepackage{mathrsfs}
\usepackage[displaymath,mathlines]{lineno}
\usepackage{color}
\usepackage{tabulary}
\usepackage{multirow}
\usepackage{pbox}
\usepackage{multicol}
\usepackage{lipsum}
\usepackage{float}

\usepackage{dcolumn}
\usepackage{booktabs}
\newcolumntype{d}[1]{D..{#1}}

\usepackage{amsthm}
\usepackage{relsize}
\usepackage{lipsum}
\usepackage{epstopdf}
\usepackage{mathtools}
\usepackage{calc}
\usepackage{makecell}
\usepackage{caption}
\usepackage{subcaption}
\usepackage{algorithm}
\usepackage{algpseudocode}
\usepackage{tabularx}
\makeatletter
\newcommand{\multiline}[1]{%
  \begin{tabularx}{\dimexpr\linewidth-\ALG@thistlm}[t]{@{}X@{}}
    #1
  \end{tabularx}
}
\makeatother

\algrenewcommand\algorithmicforall{\textbf{foreach}}
\algrenewcommand\algorithmicindent{.8em}
\usepackage{soul}

\makeatletter
\newcommand{\doublewidetilde}[1]{{%
		\mathpalette\double@widetilde{#1}}}
\newcommand{\double@widetilde}[2]{%
		\sbox\z@{$\m@th#1\widetilde{#2}$}%
		\ht\z@=.5\ht\z@
		\widetilde{\box\z@}}
\makeatother

\newtheorem{lemma}{Lemma}
\newtheorem{corollary}{Corollary}

\def\di{\displaystyle}
\newcommand{\tron}[1]{\left(\di #1 \right)} 
\newcommand{\abs}[1]{\left|#1\right|}
\newcommand{\chuan}[1]{\left\lVert #1 \right\rVert}
\newcommand{\vuong}[1]{\left[ #1 \right]} 
\newcommand{\nhon}[1]{\left\{ #1 \right\}} 
\DeclareMathOperator{\Tr}{tr}

\newcommand{\oke}[1]{{{#1}}}
 
\begin{document}
\title{\huge Improved Differential Evolution for Enhancing the Aggregated Channel Estimation of RIS-Aided Cell-Free Massive MIMO}

\author{Trinh Van Chien, Nguyen Hoang Viet,  
Symeon Chatzinotas, \textit{IEEE Fellow}, and  Lajos Hanzo, \textit{IEEE Life Fellow}
  \vspace{-0.8cm}

\thanks{This research is funded by the Vietnam Ministry of Education and Training under project number B2025-BKA04 for Trinh Van Chien. L. Hanzo would like to acknowledge the financial support of the Engineering and Physical Sciences Research Council (EPSRC) projects under grant EP/Y026721/1, EP/W032635/1 and EP/X04047X/1 as well as of the European Research Council's Advanced Fellow Grant QuantCom (Grant No. 789028).}

\thanks{T. V. Chien and N. H. Viet are with the School of Information and Communications Technology, Hanoi University of Science and Technology, Hanoi 100000, Vietnam (e-mail: chientv@soict.hust.edu.vn and nguyenhoangviet.hsgs@gmail.com). S.
Chatzinotas is with the Interdisciplinary Centre for Security, Reliability and
Trust (SnT), University of Luxembourg, L-1855 Luxembourg, Luxembourg
(email: symeon.chatzinotas@uni.lu). L. Hanzo is with the Department of Electronics and Computer
Science, University of Southampton, Southampton SO17 1BJ, U.K. (email: lh@ecs.soton.ac.uk)}

}



\maketitle

\begin{abstract}
Cell-Free Massive multiple-input multiple-output (MIMO) systems are investigated with the support of a reconfigurable intelligent surface (RIS). The RIS phase shifts are designed for improved channel estimation in the presence of spatial correlation. Specifically, we formulate the channel estimate and estimation error expressions using linear minimum mean square error (LMMSE) estimation for the aggregated channels. An optimization problem is then formulated to minimize the average normalized mean square error (NMSE) subject to practical phase shift constraints. To circumvent the problem of inherent nonconvexity, we then conceive an enhanced version of the differential evolution algorithm that is capable of avoiding local minima by introducing an augmentation operator applied to some high-performing {Diffential Evolution (DE)} individuals. Numerical results indicate that our proposed algorithm can significantly improve the channel estimation quality of the state-of-the-art benchmarks. 
\end{abstract}

\begin{IEEEkeywords}
 Cell-free massive MIMO, reconfigurable intelligent surface, differential evolution
\end{IEEEkeywords}

\vspace{-0.5cm}
\section{Introduction}
Cell-Free Massive MIMO systems constitute a promising technology for next-generation (NG) wireless networks. In these system architectures, a large number of access points (APs) equipped with multiple antennas are spread across the coverage area to serve the users. The spatial diversity gain of Cell-Free Massive MIMO technology results in significant performance enhancements \cite{ngo2017cell}. Despite this appeal, Cell-Free Massive MIMO systems still encounter challenges such as pilot contamination and the support of remote users. The detector of Cell-Free Massive MIMO systems relies on instantaneous channel state information (CSI), which is obtained by the access points (APs) from the uplink pilot training. With many users admitted to the network, having \oke{mutually orthogonal pilots} for each user is impractical, partly because the bandwidth available limits the number of orthogonal pilot sequences. However even if all pilots were orthogonal, the orthogonality may be eroded by the limited coherence interval. Therefore, several users may share the same pilot signals that introduce interference \cite{jose2011pilot}.

Integrating RISs into Cell-Free Massive MIMO systems can tackle these limitations by dynamically controlling the radio signals through reflection, refraction or absorption. Thanks to the reconfigurable phase shift elements, RIS beneficially ameliorates the incoming waves to obtain a constructive signal combination at the receiver \cite{an2020optimal}. RIS-aided Cell-Free Massive MIMO systems offer various advantages for NG networks by enhancing signal strength, reducing mutual interference, and extending system coverage, even in challenging propagation environments \cite{shi2024ris}. {Acquiring the perfect CSI in RIS-aided networks is nontrivial due to the passive nature of elements \cite{xu2023reconfiguring}. The pilot overhead required for channel estimation may become excessive for the support of the RIS if the cascaded two-hop channels are estimated separately \cite{xu2022channel}. Thus, the aggregated channels involved in both links are jointly estimated without extra pilot overhead  \cite{van2021reconfigurable, li2023low, an2022joint}. Accurate phase shift designs are paramount and require sophisticated techniques capable of operating in the face of uncertainty.} {Previous research in \cite{yan2023passive,jiang2023two} and the reference therein have demonstrated the benefits of deterministic optimization based on the first- and second-order derivatives of the objective function and constraints. We emphasize that it is nontrivial to compute these derivatives and update the solution along with imperfect channel state information and in the presence of spatial correlation.}  

{As a sub-field of artificial intelligence and optimization, evolutionary algorithms (EAs) have successfully been harnessed for solving challenging nonconvex problems, regardless of the nature of constraints and variables. For example, EAs have been successfully harnessed for maximizing the sum rate \cite{dai2023two} or for minimizing the number of RISs \cite{huang2022placement}. In the family of EAs, Differential Evolution (DE) excels in terms of providing an effective mechanism of tackling continuous-valued optimization problems. For example, if in our system there are more users than the number of unique orthogonal pilot sequences, because their length is limited in a given bandwidth, then pilot contamination is encountered. Thanks to the raw power of DE, this phenomenon can be mitigated by appropriately adjusting its evolutionary operators.} However, the conventional DE still has some issues, since its performance hinges on the specific choice of its hyperparameters. Since these parameters do not readily lend themselves to nimble adaptation, the fixed values may not be effective in the face of variations in the input factors. The standard DE often faces the challenge of getting stuck at a local optimum owing to its premature convergence, resulting in suboptimality. 

Against this backdrop, we study the aggregated channel estimation of  RIS-aided Cell-Free Massive MIMO systems in the presence of spatial correlation. The main contributions are concisely summarized as follows: 
\begin{itemize}
    \item We derive the closed-form expression of the channel estimate and estimation error by using the LMMSE estimator in the presence of spatial correlation at both the transmitter and receiver. Furthermore, we introduce the \oke{LMMSE} of the channel estimate for each user as a criterion to assess the channel estimation quality;
    \item We formulate the optimization problem of minimizing the average NMSE of the overall network subject to specific phase shift constraints. Because the problem is NP-hard, we propose an improved DE-based algorithm for obtaining a high-quality solution in polynomial time. Our proposed algorithm is capable of escaping local minima, hence preventing premature convergence; 
    \item Numerical results demonstrate that the proposed approach enhances the channel estimation quality by up to \oke{33\%} compared to the canonical DE.
\end{itemize}

\textit{Notation}: Matrices and vectors are denoted by bold letters. The notation $\text{diag}(\mathbf{x})$ is the diagonal matrix with the elements of vector $\mathbf{x}$ on the main diagonal. The symbols $\Tr\tron{\cdot}$ and $\chuan{\cdot}$ are trace and norm operators, respectively. Meanwhile, the superscript $\tron{\cdot}^H$ denote the Hermitian transpose. $\mathcal{CN}\tron{\cdot, \cdot}$ and $\mathcal{U}\tron{\vuong{a,b}}$ denote the circularly symmetric Gaussian distribution and the uniform distribution in $\vuong{a,b}$. The expectation of a random variable is $\mathbb{E}\nhon{\cdot}$.  Finally, $\mathbf{I}_N$ is the $n$-dimensional identity matrix.

\vspace{-0.2cm}
\section{System Model, Pilot Training, and Aggregated Channel Estimation}
This section presents the system considered associated with an arbitrary pilot reuse factor. We then formulate an optimization problem for improving the quality of channel estimation. 
\vspace{-0.5cm}
\subsection{Channel Model}
Let us consider a system consisting of $L$ APs each equipped with $M$ antennas serving $K$ single-antenna users. We assume that all the APs and users are randomly distributed in the coverage area riddled with obstacles. Thus, the network is assisted by an RIS that has $N$ controllable reflecting elements. 
The phase shift matrix of the RIS is $\mathbf{\Phi} = \mathrm{diag} \{e^{i \theta_{1}}, \ldots , e^{i\theta_{N}} \} \in \mathbb{C}^{N\times N}$, where $\theta_{k} \in \vuong{-\pi,\pi},\forall k = \overline{1,N}$  is the phase shift applied by the $k$-th element of the RIS. The direct channel between AP~$m$ and user~$k$, denoted by $\mathbf{g}_{mk}\in \mathbb{C}^M$, follows the Rayleigh fading distribution, i.e., we have $\mathbf{g}_{mk} \sim \mathcal{CN}(0, \mathbf{G}_{mk})$, where $\mathbf{G}_{mk} \in \mathbb{C}^{M \times M}$ is the covariance matrix that describes the corresponding spatial correlation and the large-scale fading.\footnote{{Cell-free networks are often deployed in, for example, urban areas associated with many scatterers. Hence, the Rayleigh fading distribution is well suited for such rich scattering propagation environments.}} Thanks to the support of the RIS, the channel between APs and the RIS, as well as the RIS and users, is modelled by the Rician distribution. {We
denote the channel between AP $m$ and the RIS by $\mathbf{H}_m \in \mathbb{C}^{M \times N}$, which can be formulated as \cite{shiu2000fading} $\mathbf{H}_m = \overline{\mathbf{H}}_m + \mathbf{R}_{m,\mathrm{AP}}^{1/2}\mathbf{X}_m\mathbf{R}_{m,\mathrm{RIS}}^{1/2}$,
where $\overline{\mathbf{H}}_m$ denotes the LoS channel. The elements of $\mathbf{X}_m$ are  identically and independently distributed (iid) as $\mathcal{CN}\tron{0,1}$ with $\mathbf{R}_{m,\mathrm{AP}} \in \mathbb{C}^{M\times M}$ and $\mathbf{R}_{m,\mathrm{RIS}} \in \mathbb{C}^{N\times N}$ including the spatial correlation and large-scale fading effects.} $\mathbf{z}_{k} \in \mathbb{C}^N$ stands for the channel between the RIS and user~$k$, given as $\mathbf{z}_{k} \sim \mathcal{CN}(\overline{\mathbf{z}}_{k}, \mathbf{R}_{k})$, where $\overline{\mathbf{z}}_{k}$ involves the LoS components and $\mathbf{R}_k$ denotes the spatial correlation and large-scale fading. 
\vspace{-0.2cm}
\subsection{Pilot Training with an Arbitrary Pilot Reuse Factor}
{We assume that the pilot signals received in the uplink from the users are formed as a unit vector in the $\tau_p$-dimensions complex field. Let
$\pmb{\phi}_{k} \in \mathbb{C}^{\tau_p} $ associated with $\chuan{\pmb{\phi_{k}}} =1 $ denote the pilot signal allocated to user $k$. We denote the set of indices of the users that share the same pilot as user $k$ by $\mathcal{P}_{k}$. The pilot signals are mutually orthogonal, in which $\pmb{\phi}_{k'}^H \pmb{\phi}_{k} = 
    1$ if $k' \in \mathcal{P}_{k}$. Otherwise, $\pmb{\phi}_{k'}^H \pmb{\phi}_{k} = 
    0$.  If $\tau_p < K$, the pilot contamination appears as coherent interference.}
 The aggregated channel between user~$k$ and AP~$m$ is defined as $\mathbf{u}_{mk} = \mathbf{g}_{mk} + \mathbf{H}_m\mathbf{\Phi} \mathbf{z}_{k}$, which is encapsulated in the received training signal as
\begin{equation}
    \mathbf{Y}_{pm} = \sqrt{p\tau_p} \sum\nolimits_{k=1}^K \mathbf{u}_{mk}\pmb{\phi}_{k}^H + \mathbf{W}_{pm},
\end{equation}
where $p$ is the normalized signal-to-noise ratio (SNR) of each pilot symbol, and $\mathbf{W}_{pm}$ is the additive white Gaussian noise. Upon projecting $\mathbf{Y}_{pm}$ onto $\pmb{\phi}_{k}^H$, one can obtain
\begin{equation}
    \mathbf{y}_{pmk} = \sqrt{p\tau_p}\sum\nolimits_{k' \in \mathcal{P}_{k}}\mathbf{u}_{mk'} + \mathbf{w}_{pmk},  \label{receiveSignal}
\end{equation}
where $\mathbf{w}_{pmk} \sim \mathcal{CN}(0,\mathbf{I}_M)$.
The aggregated channel between AP~$m$ and  user~$k$ impinging through the RIS is $\mathbf{u}_{mk} = \mathbf{g}_{mk} + \mathbf{H}_m\mathbf{\Phi} \mathbf{z}_{k}.$
Based on the signal observation in \eqref{receiveSignal} and applying the linear MMSE estimator of \cite{kay1993fundamentals}, the estimated version of the aggregated channel $\mathbf{u}_{mk}$ can be formulated as 
\begin{equation} 
    \hat{\mathbf{u}}_{mk} = \overline{\mathbf{u}}_{mk} + \mathbb{E}\{\widetilde{\mathbf{u}}_{mk}\widetilde{\mathbf{y}}_{pmk}^H\} \big( \mathbb{E}\{\widetilde{\mathbf{y}}_{pmk}\widetilde{\mathbf{y}}_{pmk}^H\} \big)^{-1} \widetilde{\mathbf{y}}_{pmk}, \label{estimatedchannel}
\end{equation}
where $\overline{\mathbf{u}}_{mk} = \mathbb{E}\{ \mathbf{u}_{mk} \}$, $\widetilde{\mathbf{u}}_{mk} = \mathbf{u}_{mk} - \overline{\mathbf{u}}_{mk}$, and  $\tilde{\mathbf{y}}_{pmk} = \mathbf{y}_{pmk} - \mathbb{E}\{ \mathbf{y}_{pmk} \}$. The closed-form expression of the channel estimate $\hat{\mathbf{u}}_{mk}$ is obtained by computing the expectations in \eqref{estimatedchannel}, as shown in Lemma~\ref{Lemma:ChannelEst}.
\begin{lemma} \label{Lemma:ChannelEst}
If the LoS components are not blocked completely and the LMMSE estimator is used, then the aggregated channel estimate $\hat{\mathbf{u}}_{mk}, \forall m,k,$  can be expressed in closed form as
\begin{equation} \label{eq:ClosedformEst}
\hat{\mathbf{u}}_{mk} = \overline{\mathbf{H}}_m\mathbf{\Phi} \overline{\mathbf{z}}_{k}  + \mathbf{\Gamma}_{mk}\mathbf{\Psi}_{mk}^{-1} \widetilde{\mathbf{y}}_{pmk},
\end{equation}
where 
$\mathbf{\Delta}_{mk} = \mathbf{G}_{mk} + \overline{\mathbf{H}}_m\mathbf{\Phi}\mathbf{R}_{k}\mathbf{\Phi}^H\overline{\mathbf{H}}_m^H + \Tr(\mathbf{R}_{m,\mathrm{RIS}}\mathbf{\Phi}\mathbf{R}_{k}\mathbf{\Phi}^H)\mathbf{R}_{m,\mathrm{AP}}, \mathbf{\Gamma}_{mk} =  \sqrt{p\tau_p} \left(\mathbf{\Delta}_{mk} + \Tr\left(\mathbf{R}_{m,\mathrm{RIS}}\mathbf{\Phi}\overline{\mathbf{z}}_{k}\left(\sum\nolimits_{k' \in \mathcal{P}_k} \overline{\mathbf{z}}_{k'}^H\right)\mathbf{\Phi}^H\right)\mathbf{R}_{m,\mathrm{AP}}\right), \mathbf{\Psi}_{mk} =p\tau_p\Tr(\mathbf{R}_{m,\mathrm{RIS}}\mathbf{\Phi}(\sum\nolimits_{k' \in \mathcal{P}_k} \overline{\mathbf{z}}_{k'})(\sum\nolimits_{k' \in \mathcal{P}_k} \overline{\mathbf{z}}_{k'}^H)\mathbf{\Phi}^H)\mathbf{R}_{m,\mathrm{AP}}+ p\tau_p\sum_{k' \in \mathcal{P}_k} \mathbf{\Delta}_{mk'} + \mathbf{I}_M$, $\mathbf{\widetilde{y}}_{pmk} = \sqrt{p\tau_p}\sum_{k' \in \mathcal{P}_k}\widetilde{\mathbf{u}}_{mk'} + \mathbf{w}_{pmk}$.
 In addition, the channel estimation error defined as $\mathbf{e}_{mk} = \mathbf{u}_{mk} - \hat{\mathbf{u}}_{mk}$ has zero mean and the covariance matrix equals to
\begin{equation}\label{errorvariance}
 \mathbb{E}\{\mathbf{e}_{mk}\mathbf{e}_{mk}^H\} = \mathbf{\Delta}_{mk} +  \alpha_{mk}\mathbf{R}_{m,\mathrm{AP}} - \mathbf{\Gamma}_{mk}\mathbf{\Psi}_{mk}^{-1}\mathbf{\Gamma}_{mk}^H,
\end{equation}
where $\alpha_{mk} =\Tr(\mathbf{R}_{m,\mathrm{RIS}}\mathbf{\Phi}\overline{\mathbf{z}}_{k}\overline{\mathbf{z}}_{k}^H\mathbf{\Phi}^H)$.
\end{lemma}

\begin{proof} 
See Appendix B.
\end{proof}
{In contrast to the channel estimate in \eqref{estimatedchannel}, the closed-form expressions in \eqref{eq:ClosedformEst} and \eqref{errorvariance} explicitly unveil the influences of spatial correlation, pilot reuse pattern, and phase shift design on the channel estimation quality. We also stress that the aggregated channel estimate in \eqref{eq:ClosedformEst} does not require any additional pilot training overhead compared to that of the conventional Cell-Free Massive MIMO systems. Having highly accurate channel estimation is of paramount importance in boosting the achievable rate of the uplink and downlink data transmission, especially in the time division duplex mode.} 
\vspace{-0.2cm}
\subsection{Problem formulation}
We first denote the normalized mean square error (NMSE) of the estimated channel between user $k$ and AP~$m$ as
\begin{equation}\label{NMSEformula}
    \mathrm{NMSE}_{mk} =\dfrac{\mathbb{E}\nhon{\chuan{\mathbf{e}_{mk}}^2}}{\mathbb{E}\nhon{\chuan{\mathbf{u}_{mk}}^2}} \stackrel{(a)}{=}  \dfrac{\mathbb{E}\nhon{\chuan{\mathbf{e}_{mk}}^2}}{\mathbb{E}\nhon{\|\hat{\mathbf{u}}_{mk}\|^2} + \mathbb{E}\nhon{\chuan{\mathbf{e}_{mk}}^2}}, 
\end{equation}
where $(a)$ is obtained by exploiting the independence of the channel estimate and the resultant estimation error. The second equality of \eqref{NMSEformula} indicates that the NMSE is constrained to the interval $\vuong{0,1}$. {By exploiting Lemma~\ref{Lemma:ChannelEst}, one can formulate the closed-form expression of the NMSE as 
\begin{align} \label{NMSEformulav1}
&\mathrm{NMSE}_{mk} = 
\dfrac{\Tr\tron{\mathbf{\Delta}_{mk} +  \alpha_{mk}\mathbf{R}_{m,\mathrm{AP}} - \mathbf{\Gamma}_{mk}\mathbf{\Psi}_{mk}^{-1}\mathbf{\Gamma}_{mk}^H}}{\Tr\tron{\overline{\mathbf{H}}_m\mathbf{\Phi}\overline{\mathbf{z}}_k\overline{\mathbf{z}}_k^H\mathbf{\Phi}^H \overline{\mathbf{H}}_m^H + \mathbf{\Delta}_{mk} +  \alpha_{mk}\mathbf{R}_{m,\mathrm{AP}}}} =  1 - \notag\\
&\dfrac{\Tr\tron{\overline{\mathbf{H}}_m\mathbf{\Phi}\overline{\mathbf{z}}_k\overline{\mathbf{z}}_k^H\mathbf{\Phi}^H \overline{\mathbf{H}}_m^H + \mathbf{\Gamma}_{mk}\mathbf{\Psi}_{mk}^{-1}\mathbf{\Gamma}_{mk}^H}}{\Tr\tron{\overline{\mathbf{H}}_m\mathbf{\Phi}\overline{\mathbf{z}}_k\overline{\mathbf{z}}_k^H\mathbf{\Phi}^H \overline{\mathbf{H}}_m^H + \mathbf{\Delta}_{mk} +  \Tr\tron{\mathbf{R}_{m,\mathrm{RIS}}\mathbf{\Phi}\overline{\mathbf{z}}_{k}\overline{\mathbf{z}}_{k}^H\mathbf{\Phi}^H}\mathbf{R}_{m,\mathrm{AP}}}}.
\end{align} 
Observe from \eqref{NMSEformulav1} that the NMSE can be arbitrarily small if orthogonal pilot signals are harnessed and the pilot power $p$ is sufficiently high. In the case of pilot reuse, the NMSE will no longer tend to zero. We stress that the NMSE obtained in \eqref{NMSEformulav1} is a generic version of the previous work  where was only applied for correlated Rayleigh fading \cite{van2021reconfigurable} or without the presence of pilot contamination \cite{zhi2022two}.} Consequently, we propose to design the phase shifts by minimizing the average NMSE of the overall network as follows:
\begin{alignat}{2}\label{mainprob}
    &\underset{\{\theta_1,\theta_2,\ldots,\theta_N\}}{\mathrm{minimize}} && \,  \frac{1}{L   K}\sum\nolimits_{m=1}^L \sum\nolimits_{k=1}^K \text{NMSE}_{mk}\\
    &\text{subject to} && -\pi \le \theta_{n} \le \pi, \forall \, 1 \le n \le N. \notag
\end{alignat}
Problem~\eqref{mainprob} is non-convex and it is nontrivial to obtain the first and second derivatives of the objective function under spatial correlation. Hence, solving \eqref{mainprob} is challenging.
\section{Improved DE-based Solution}
\subsection{DE Relying on Augmentation Operator}
To tackle the complex optimization of the NMSE in \eqref{NMSEformulav1} and the non-convexity of problem \eqref{mainprob}, we conceive an improved version of DE by integrating it with an augmentation strategy, termed as Augmentation DE (ADE).
To elaborate, initially, a population of individuals is generated, which is then evolved through consecutive generations. The mutation and crossover operators are harnessed for generating new individuals in each generation. Then the fitness of each new offspring is evaluated and compared directly to that of its parent. The specific individuals having a lower average NMSE are then passed on to the next generation. When this happens, the parent individual will be stored in an archive set $\mathcal{A}$ for later use in order to maintain search diversity. The current solution of problem \eqref{mainprob} is the particular individual having the best fitness value, i.e. the lowest average NMSE. \oke{Again, to avoid getting stuck at a local optimum as the conventional DE, we harness an augmentation operator to improve the current solution. Additionally, unlike in the conventional DE, the hyperparameters in the mutation and crossover operators are dynamically adjusted. This method enables ADE to update the solution based on the search behavior after each iteration, resulting in improved solutions in subsequent generations.} A detailed description of Algorithm~\ref{alg:SDE} is provided below.
\begin{algorithm}
    \caption{Augmentation Differential Evolution (ADE)}
    \label{alg:SDE}
\hspace*{\algorithmicindent}\textbf{Input:} The channel $\mathbf{g}_{mk}, \mathbf{H}_m, \mathbf{z}_k$ with $k = \overline{1,K}, m = \overline{1,L}$; the maximum number of generations $G_{\text{MAX}}$; the augmentation search condition $\varepsilon$ and augmentation parameter $\beta$. \\
    \hspace*{\algorithmicindent}\textbf{Output:} The phase shift matrix $\mathbf{\Phi}$. 
    \begin{algorithmic}[1]
        \State Set the generation index $G \gets 1$. Randomly initialize a population $\mathcal{P}^{(1)}$ with $I$ individuals and set archive set $\mathcal{A} \gets \mathcal{P}^{(1)}$.
        \While{$G < G_{\text{MAX}}$} 
            \State Initialize the next generation $\mathcal{Q}^{(G)} \gets \varnothing$. 
            \For{$i = 1 : I$}
                \State \multiline{%
                Generate $\mathsf{F}^{(iG)}$ and $\mathsf{CR}^{(iG)}$
                as the SHADE method \cite{tanabe2013success}.
                Create mutant vector $\mathbf{u}^{(iG)}$ using DE/$p$best/1 strategy in \eqref{mutop}.
                Create trial vector $\pmb{\omega}^{(iG)}$ using \eqref{crosop}.
                }
                \If {$f(\pmb{\omega}^{(iG)}) \le f(\mathbf{x}^{(iG)})$}
                    \State $\mathcal{Q}^{(G)} \gets \mathcal{Q}^{(G)} \cup \pmb{\omega}^{(iG)}$, $\mathcal{A}_j \gets \mathbf{x}^{(iG)}$
                \Else 
                    \State $\mathcal{Q}^{(G)} \gets \mathcal{Q}^{(G)} \cup \mathbf{x}^{(iG)}$
                \EndIf 
                \If {$f^{(\text{best}G)} - f^{(\text{best}Q)} < \varepsilon$}
                    \For {$j = 1 : \lambda$}
                        \State Create augmented vector $\widetilde{\pmb{\omega}}^{(jG)}$ using \eqref{shiftop}.
                        \If {$f(\widetilde{\pmb{\omega}}^{(jG)})< f(\pmb{\omega}^{(jG)})$ }
                            \State $\mathcal{Q}^{(G)}_j \gets \widetilde{\pmb{\omega}}^{(jG)}$
                        \ElsIf {$f(\widetilde{\pmb{\omega}}^{(jG)})< f\tron{\mathcal{A}_j}$}
                            \State $\mathcal{A}_j \gets \widetilde{\pmb{\omega}}^{(jG)}$
                        \EndIf
                    \EndFor
                \EndIf
            \EndFor
            \State $\mathcal{P}^{(G)} \gets \mathcal{Q}^{(G)}$, $G \gets G + 1$.
        \EndWhile
    \end{algorithmic}
\end{algorithm}

\textit{1) Solution representation and population initialization:} The population $\mathcal{P}^{(1)}$ contains $I$ individuals, which are randomly generated. The $i$-th individual of the $G$-th generation is denoted by $\mathbf{x}^{(iG)} \in \mathbb{R}^N$, whose elements range from $-\pi$ to $\pi$. 

\textit{2) Mutation:} In the $G$-th generation, a mutant vector $\mathbf{u}^{(iG)} \in \mathbb{R}^N$ is created corresponding to its parent $\mathbf{x}^{(iG)}$. In this paper, we employ the popular DE/$p$best/1 strategy to create a mutant vector, which is formulated as
\begin{equation}\label{mutop}
    \mathbf{u}^{(iG)} = \mathbf{x}^{(p\text{best}G)} + \mathsf{F}^{(iG)}(\mathbf{x}^{(r_1G)} - \mathbf{x}^{(r_2G)}),
\end{equation}
where $\mathbf{x}^{(p\text{best}G)}$ is chosen from the $\lfloor I \times p \rfloor$ individuals with the best fitness values in the current population $\mathcal{P}^{(G)}$, whereas $\mathbf{x}^{(r_1G)}$ and $\mathbf{x}^{(r_2G)}$ is selected randomly from $\mathcal{P}^{(G)}$ and $\mathcal{P}^{(G)} \cup \mathcal{A}$, respectively. The term $\mathbf{x}^{(p\text{best}G)}$ is employed for ensuring that the search behaviour is not be dominated by any particular individual. Meanwhile,  the scale factor $\mathsf{F}^{(iG)}$ of the $i$-th individual will be adapted according to the specific nature of the search behaviour instead of using a fixed value throughout all generations, as proposed in \cite{tanabe2013success}. Nonetheless, the update in \eqref{mutop} may result in having some elements of $\mathbf{u}^{(iG)}$ to fall outside the feasible domain. To tackle this risk, we adjust the mutant vector formulated mathematically as 
\begin{equation}
    [\mathbf{u}^{(iG)}]_t = 
    \begin{cases}
        (1 + [\mathbf{x}^{(iG)}]_t)/2, & \text{if } [\mathbf{u}^{(iG)}]_t > 1, \\
        (-1 + [\mathbf{x}^{(iG)}]_t)/2, & \text{if } [\mathbf{u}^{(iG)}]_t <- 1, 
    \end{cases}
\end{equation} 
where $[\mathbf{a}]_t$ is the $t$-th element of  vector $\mathbf{a}$. 

\textit{3) Crossover:} After the mutation, $\mathbf{u}^{(iG)}$ is combined with its parents to generate a trial vector $\pmb{\omega}^{(iG)} \in \mathbb{R}^N$, which is
\begin{equation}\label{crosop}
    [\pmb{\omega}^{(iG)}]_t = \begin{cases}
        [\mathbf{u}^{(iG)}]_t, &\text{if } \mathcal{U}\vuong{0,1} \le \mathsf{CR}^{(iG)} \mbox{ or } t = t_{\mathrm{rand}}, \\
        [\mathbf{x}^{(iG)}]_t, &\text{otherwise,}
    \end{cases}
\end{equation}
where $\mathcal{U}\vuong{0,1}$ represents a
uniform random number range between $0$ and $1$,  $\mathsf{CR}^{(iG)}$ is the crossover rate, and $t_{rand}$ is chosen randomly between $1$ and $N$ to ensure that $\pmb{\omega}^{(iG)}$ has at least one element from the mutant vector $\mathbf{u}^{(iG)}$. 

\textit{4) Selection:} 
The trial vector $\pmb{\omega}^{(iG)}$ will replace $\mathbf{x}^{(iG)}$ in the next generation, if it has better fitness value. Meanwhile, $\mathbf{x}^{(iG)}$ will replace the $i$-th member in the archive set $\mathcal{A}$ for ensuring that $\mathcal{P}^{(G)}$ and $\mathcal{A}$ have the same cardinality.

\textit{5) Dispensing with Augmentation:} The conventional DE algorithm has a major drawback. Explicitly, it tends to promptly evolve all individuals in the population towards a local minimum without searching for a better solution. To circumvent the problem of premature convergence,  we introduce a local search operator, where all phase shift elements of each selected individual will be augmented by a random parameter to ensure that some individuals escape from the local minimum. In particular, we modify the first $\lambda$ ($\lambda \in \vuong{1,I}$) members of the population if $\abs{f^{(\text{best}Q)} - f^{(\text{best}G)}} < \varepsilon$, where $f^{(\text{best}G)}$ is the best fitness value obtained in the current generation, $f^{(\text{best}Q)}$ is the best fitness value after selection and $\varepsilon$ is a fixed tolerance threshold. The augmented vector $\widetilde{\pmb{\omega}}^{(jG)}$ ($j \in \vuong{1,\lambda}$), which has a minimum value of $\pmb{\omega}^{(jG)}$ is formulated as 
\begin{equation}\label{shiftop}
    \widetilde{\pmb{\omega}}^{(jG)} = \pmb{\omega}^{(jG)} + \beta  \mathbf{1}_N,
\end{equation}
where $\mathbf{1}_N$ is a vector having $N$ elements equal to $1$ and $\beta \sim \mathcal{N}(0,\sigma^2)$ has a low variance $\sigma^2$. Note that $\widetilde{\pmb{\omega}}^{(jG)}$ will replace $\pmb{\omega}^{(jG)}$ in the next generation if $\oke{f(\widetilde{\pmb{\omega}}^{(jG)}) < f(\pmb{\omega}^{(jG)})}$. Otherwise, the augmented vector $\widetilde{\pmb{\omega}}^{(jG)}$ will replace the $j$-th individual in the archive set $\mathcal{A}$. 

\textit{6) Parameter adaption:} The performance of Algorithm~\ref{alg:SDE} significantly depends on the choice of the control parameter $\mathsf{F}^{(iG)}$ in \eqref{mutop} and $\mathsf{CR}^{(iG)}$ in \eqref{crosop} because they directly affect both the mutation and crossover mechanisms. It is adjusted throughout the generations relying on the  so-called success-history-based parameter adaptation (SHADE) of \cite{tanabe2013success}.\footnote{{Thanks to intelligent computation, Algorithm~1 can potentially be extended to scenarios of multiple-RIS-aided systems by including the RIS indices as optimization variables and extending the evolutionary dimensions. Due to the limited space in this compactness, this promising research direction is left for future research.}}

\subsection{Computational Complexity}
As a function of the phase shift matrix $\mathbf{\Phi}$, the channel state information $\mathbf{u}_{mk}$,  $\mathbf{\hat{u}}_{mk}$, $\mathbf{e}_{mk}$, and the average NMSE in \eqref{NMSEformulav1} are updated accordingly. Specifically,  the update of the covariance matrix of aggregated channel $\mathbf{u}_{mk}$ has a computational complexity order of $\mathcal{O}(MN^2 + M^2N + N^3)$. By contrast, the matrices $\mathbf{\Gamma}_{mk}$ and $\mathbf{\Psi}_{mk}$ impose a complexity order of $\mathcal{O}(MN^2 + M^2N + N^3 + K)$. Hence, the  complexity of evaluating the objective function \eqref{mainprob} is $\mathcal{O}\vuong{LK(MN^2 +M^2N + N^3 +K)}$. {Regarding the phase shift optimization, the initialization imposes the computational complexity of $\mathcal{O}(IN)$. In each generation, sorting the population for use in (9) has the cost of $\mathcal{O}\vuong{I\log(I)}$, while both the mutation and crossover have the complexity of $\mathcal{O}(IN)$. The selection, augmentation search, and parameter adaptation have the cost of $\mathcal{O}(I), \mathcal{O}(\lambda I),$ and $\mathcal{O}(I)$, respectively. Thus, the computational complexity of Algorithm 1 is $\mathcal{O}[GI\log(I) + GIN + \lambda GI]$ with $G$ being the number of generations. Basically, at each iteration, ADE shares the same core stages as the Genetic Algorithm (GA) and DE, which includes mutation, crossover, selection, and population sorting. The significant distinction between the proposed algorithm and the two canonical benchmarks is that ADE activates the Augmentation Operator after the selection and updates the hyper-parameters according to the search behavior. These enhancements lead to the difference in the complexity of three algorithms considered, as summarized in Table~\ref{table:complexity}.}
\begin{table}[t]
\caption{The complexity of ADE, DE, and GA.}
\label{table:complexity}
\centering
\begin{tabular}{|c|c|} \hline
\multicolumn{1}{|l|}{Benchmark} & \multicolumn{1}{c|}{Complexity} \\ \hline
ADE                   & $\mathcal{O}[GI\log(I) + GIN + \lambda GI]$                              \\ \hline
 DE       & $\mathcal{O}[GI\log(I) + GIN]$                              \\ \hline
 GA       & $\mathcal{O}[GI\log(I) + GIN]$ \\ \hline                  
\end{tabular}
\end{table}

\section{Numerical Results}

\begin{figure*}
\begin{subfigure}[t]{0.24\textwidth}
    \includegraphics[scale=0.24]{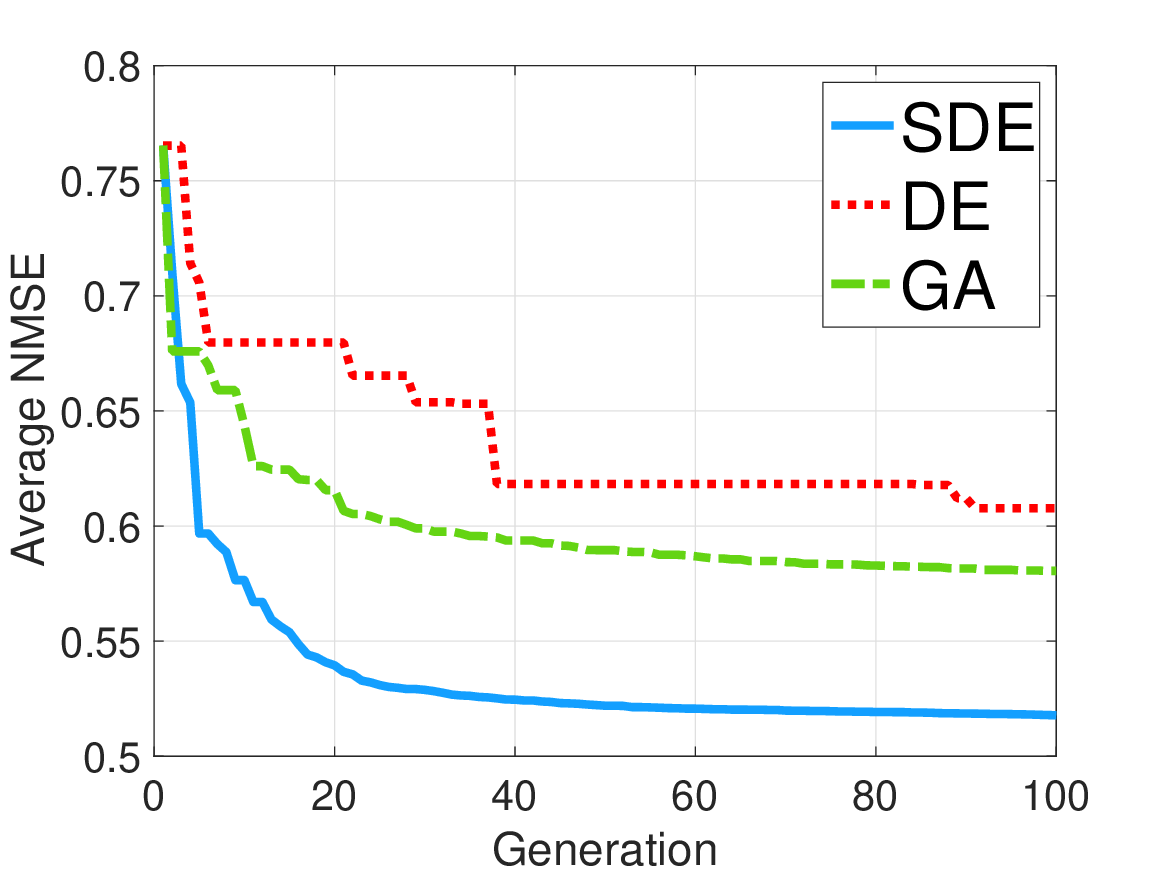}
    \caption{$\tau_p = 1,N = 100$}\label{fig:100ris-1pilot}
\end{subfigure}\hfill
\begin{subfigure}[t]{0.24\textwidth}
    \includegraphics[scale=0.24]{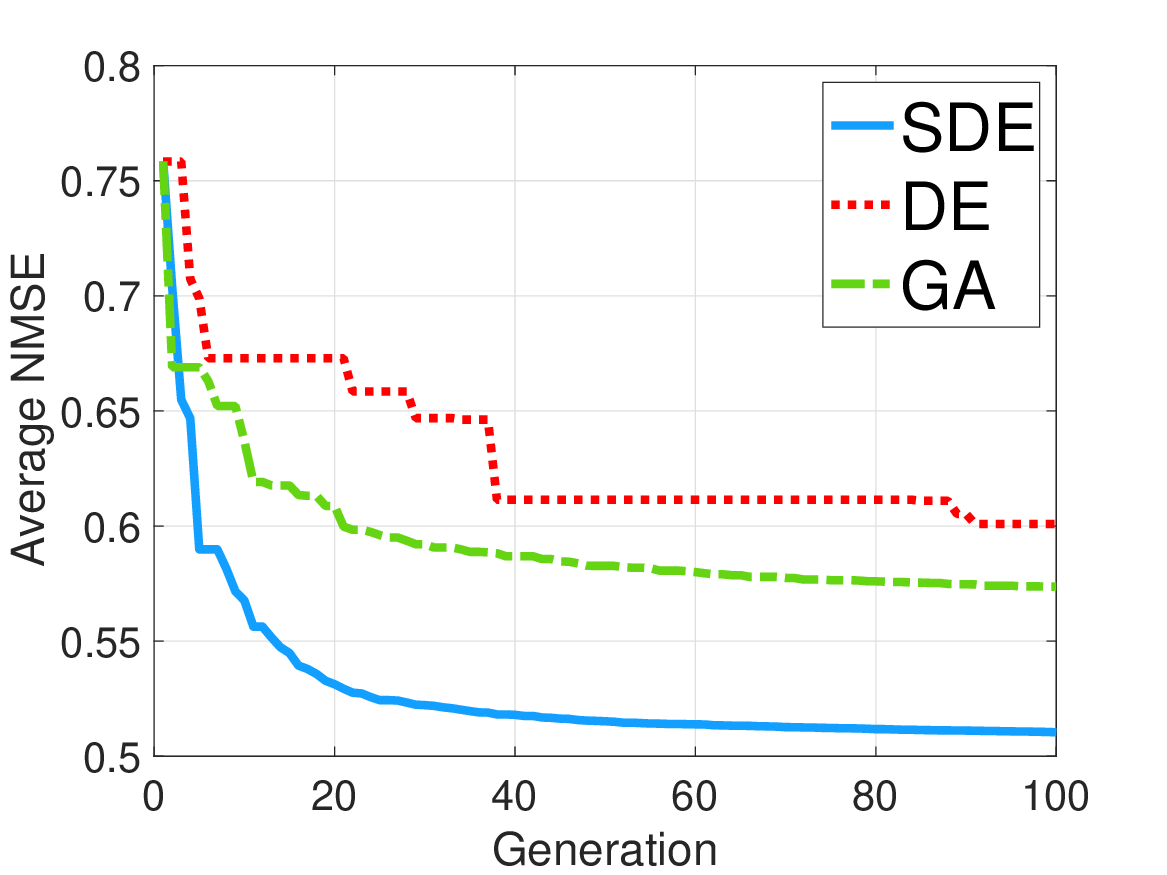}
    \caption{$\tau_p =5 , N=100$}\label{fig:100ris-5pilot}
\end{subfigure}
\hfill\begin{subfigure}[t]{0.24\textwidth}
    \includegraphics[scale=0.24]{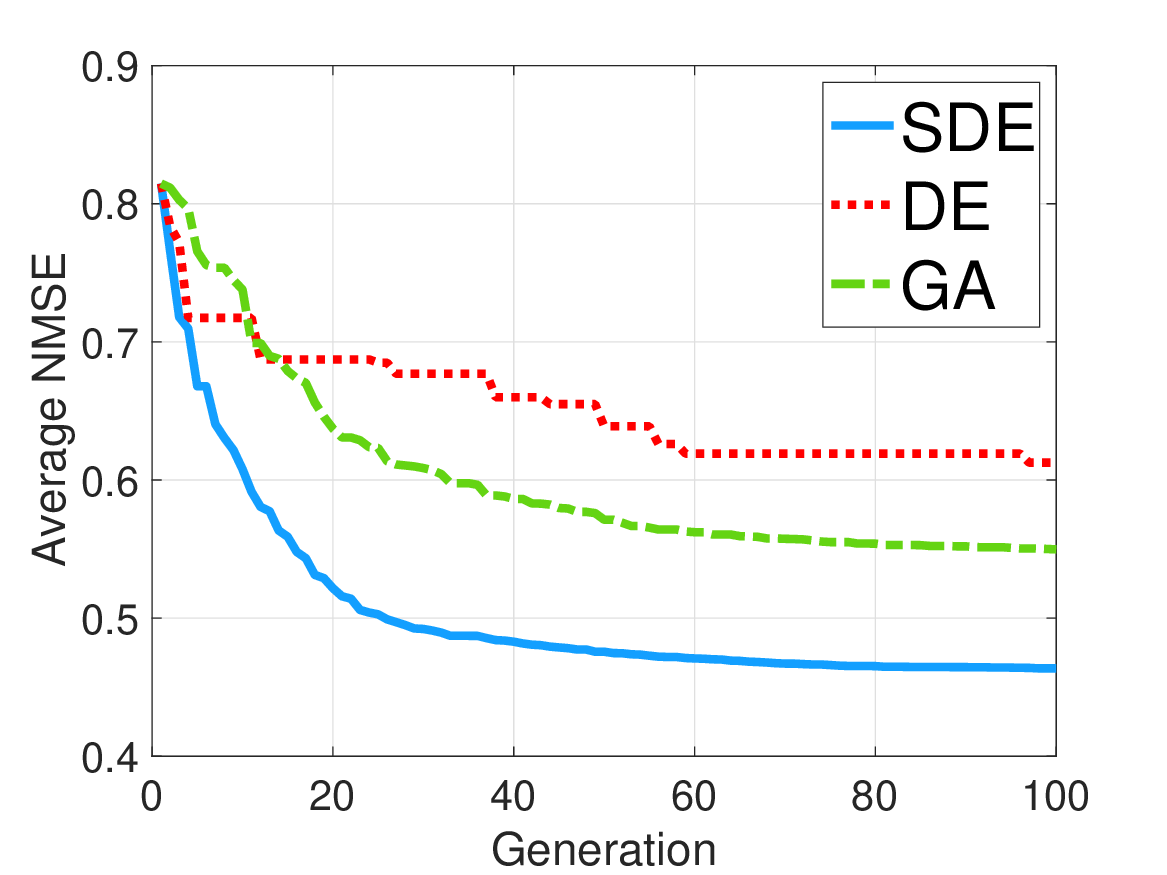}
    \caption{$\tau_p = 1, N = 256$}\label{fig:256ris-1pilot}
\end{subfigure}\hfill\begin{subfigure}[t]{0.24\textwidth}
    \includegraphics[scale=0.24]{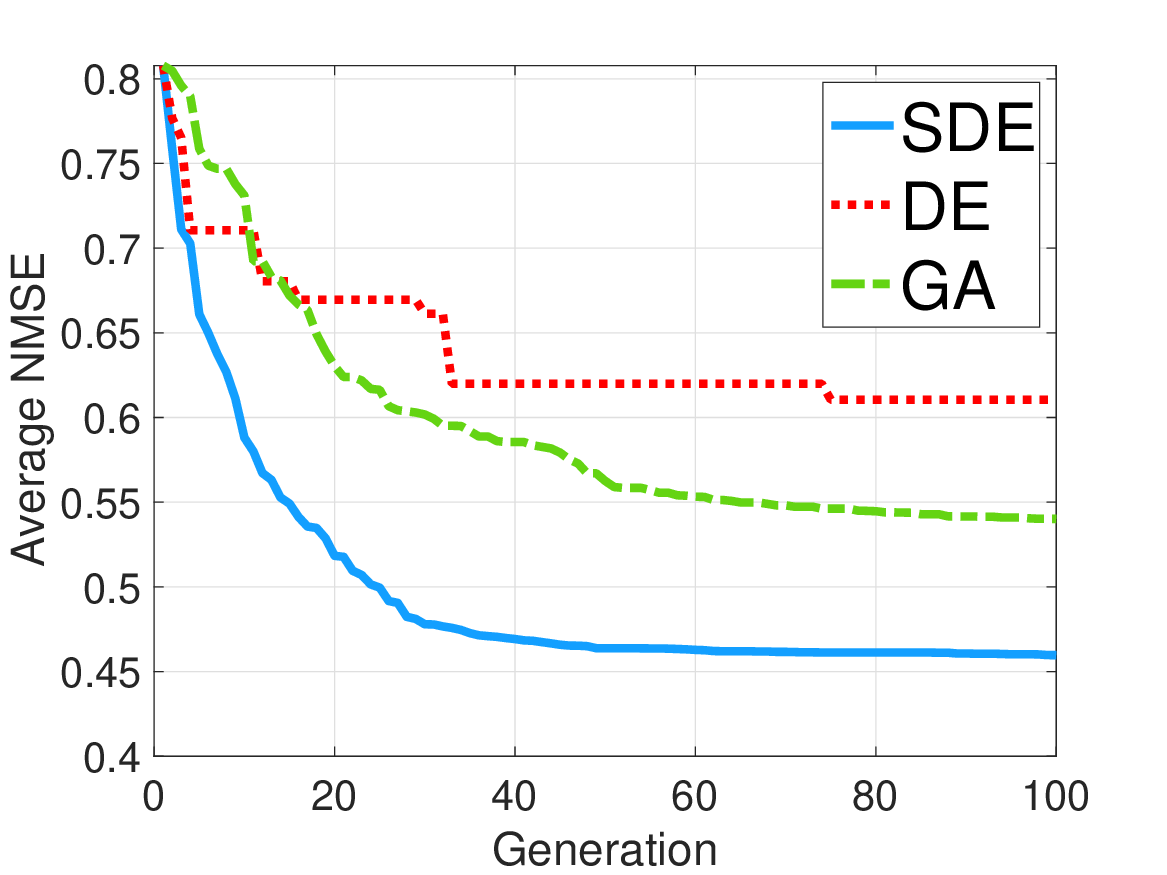}
    \caption{$\tau_p =5, N = 256$}\label{fig:256ris-5pilot}
\end{subfigure}

\caption{\label{fig:res}The convergence trend of SDE and other evolutionary algorithms, including canonical DE and GA.}
\vspace{-0.5cm}
\end{figure*}
{A system is considered that includes $40$ APs, each equipped with $16$ antennas and a single RIS, serving ten users in the square area of 1 $\text{km}^2$. These APs and users are uniformly distributed in $0.25 \times 0.25 \text{km}^2$ squares at the southwest and southeast corners, respectively, while the RIS is at the centre of the region considered. The height of the APs, RIS, and users is 15~[m], 30~[m], and 1.65~[m], respectively, while the spatial correlation matrices are defined in \cite{demir2022channel}. To simulate a harsh communication environment, it is assumed that the direct links between APs and users, i.e. $\mathbf{g}_{mk}, \forall \, m,k$, are unblocked with a probability of $\widetilde{p}$ \cite{van2021reconfigurable}. Specifically, the covariance matrix $\mathbf{G}_{mk}$ will be multiplied by a scalar $\alpha_{mk}$, where $\alpha_{mk} = 1$ with probability $\widetilde{p}$, and $\alpha_{mk} = 0$ otherwise. In our setup, $\widetilde{p} = 0.2$, meaning that, on average, each AP can directly transmit signal to two user.}
Our proposed phase shift design is compared both to the classic GA \cite{zhi2022power} and DE \cite{taha2022multi}.

In Fig.~\ref{fig:res}, we plot the convergence of all the evolutionary benchmarks considered. While both GA and DE could yield a good solution, SDE outperforms these baselines. Indeed, the results demonstrate the superior performance of our proposed algorithm, which incorporates an augmentation strategy, compared to the standard DE and GA. In the first evolutions, SDE  continuously explores the search space for finding a superior solution that DE and GA cannot achieve. Nonetheless, after many generations, all the selected individuals may correspond to a local minimum. Thus, the augmentation mechanism promotes efficiency by preventing premature convergence and continuing exploration to avoid these local minima. For completeness, two additional benchmarks, namely the equal phase shift design (EPS)  \cite{van2021reconfigurable} and random phase shift design (RPS) \cite{zhi2022power} are also harnessed for comparison. In the RPS design, each element is drawn from the uniform distribution in the interval $\vuong{-\pi,\pi}$. We evaluate the average NMSE for 1000 different realizations of the random phase shift elements, as reported in TABLE~\ref{table:res}.
\begin{table}[t]
\centering
\caption{Average NMSE with different benchmarks.}
\label{table:res}
\begin{tabular}{|c|ccccc|}
\hline
\multirow{2}{*}{\begin{tabular}[|c]{@{}c@{}}\, \\  Scenarios\end{tabular}} & \multicolumn{5}{c|}{Average NMSE}                                                                                                                                                                                                    \\ \cline{2-6} 
                          & \multicolumn{1}{c|}{SDE} & \multicolumn{1}{c|}{GA \cite{zhi2022power}}  & \multicolumn{1}{c|}{DE \cite{taha2022multi}}  & \multicolumn{1}{c|}{\begin{tabular}[c]{@{}c@{}}RPS \cite{zhi2022power}\end{tabular}} & \begin{tabular}[c]{@{}c@{}} EPS \cite{van2021reconfigurable}\end{tabular} \\ \hline
\multicolumn{1}{|c|}{\begin{tabular}[c]{@{}c@{}} $ \tau_p = 1$\\ $N = 100$ \end{tabular}}        & \multicolumn{1}{c|}{$\mathbf{\oke{0.5178}}$}  & \multicolumn{1}{c|}{\oke{$0.5805$}}  & \multicolumn{1}{c|}{\oke{$0.6077$}}  & \multicolumn{1}{c|}{\oke{$0.9096$}}                                                               & \oke{$0.9925$}                                                                          \\ \hline
\multicolumn{1}{|c|}{\begin{tabular}[c]{@{}c@{}} $ \tau_p = 5$\\ $N = 100$ \end{tabular}}        & \multicolumn{1}{c|}{\oke{$\mathbf{0.5104}$}}  & \multicolumn{1}{c|}{\oke{$0.5737$}}  & \multicolumn{1}{c|}{\oke{$0.6009$}}  & \multicolumn{1}{c|}{\oke{$0.8918$}}                                                               & \oke{$0.9856$}                                                                          \\ \hline
\multicolumn{1}{|c|}{\begin{tabular}[c]{@{}c@{}} $ \tau_p = 1$\\ $N = 256$ \end{tabular}}         & \multicolumn{1}{c|}{$\mathbf{\oke{0.4636}}$} & \multicolumn{1}{c|}{\oke{$0.5498$}} & \multicolumn{1}{c|}{\oke{$0.6126$}} & \multicolumn{1}{c|}{\oke{$0.9071$}}                                                              & \oke{$0.9940$}                                                                          \\ \hline
    \multicolumn{1}{|c|}{\begin{tabular}[c]{@{}c@{}} $ \tau_p = 5$\\ $N = 256$ \end{tabular}}         & \multicolumn{1}{c|}{$\mathbf{\oke{0.4597}}$} & \multicolumn{1}{c|}{\oke{$0.5401$}} & \multicolumn{1}{c|}{\oke{$0.6105$}} & \multicolumn{1}{c|}{\oke{$0.9003$}}                                                              & \oke{$0.9872$}                                                                         \\ \hline
\end{tabular}
\end{table}
{We now evaluate the ergodic spectral efficiency (SE) obtained with the phase shifts designed by  SDE and RPS with $N = 100$ and the pilot signal duration $\tau_p \in \{1,2,5\}$. Specifically, the uplink SE of user $k$ is formulated as 
$R_k = B(1- \tau_p/\tau_c)\log_2(1 + \gamma_{k}) \mbox{ [Mbps]}$, 
where $B$~[MHz] is the system bandwidth and the effective signal-to-interference-plus-noise ratio $\gamma_k$ is
\begin{align}
&\gamma_{k} = \dfrac{\abs{\mathsf{DS}_{uk}}^2}{\mathbb{E}\nhon{\abs{\mathsf{BU}_{uk}}^2} +  \sum_{k'=1, k' \ne k}^K\mathbb{E}\nhon{\abs{\mathsf{UI}_{uk'k}}^2} + \mathbb{E}\{\abs{\mathsf{NO}_{uk}}^2\} },
\end{align}
where $\mathsf{DS}_{uk} = \sqrt{p_u\eta_k}\mathbb{E}\nhon{\sum_{m=1}^L\mathbf{\hat{u}}_{mk}^H\mathbf{u}_{mk}}$, $\mathsf{BU}_{uk} = \sqrt{p_u\eta_k}\big(\sum_{m=1}^L\mathbf{\hat{u}}_{mk}^H\mathbf{u}_{mk} - \mathbb{E}\nhon{\sum_{m=1}^L\mathbf{\hat{u}}_{mk}^H\mathbf{u}_{mk}} \big)$, $\mathsf{UI}_{uk'k} = \sqrt{p_u\eta_{k'}}\sum_{m=1}^L\mathbf{\hat{u}}_{mk}^H\mathbf{u}_{mk'}$, and $\mathsf{NO}_{uk} = \sum_{m=1}^L \mathbf{\hat{u}}_{mk}^H\mathbf{w}_{um}$. The ergodic SE per user measured in Mbps are shown in TABLE~\ref{table:SE} for a system bandwidth $B= 10$~[MHz]. 
\begin{table}[t]
\centering
\caption{Ergodic SE per user with different benchmarks.}
\label{table:SE}
\begin{tabular}{|c|c|c|c|}
\hline
Scenarios & $\tau_p = 1$ & $\tau_p = 2$ & $\tau_p = 5$ \\ \hline
RPS       & 0.93  & 1.51  & 2.38  \\ \hline
SDE       & 1.25  & 1.93  & 2.88  \\ \hline
Gain      & 34\% & 28\% & 21\% \\ \hline
\end{tabular}
\end{table}
}

\section{Conclusions}
Improved channel estimation was conceived for RIS-assisted MIMO systems by optimizing the phase shift elements. The proposed algorithm has shown significant channel estimation error reduction. The solutions were found in polynomial time without relying on the gradient of the objective function and outperformed the state-of-the-art benchmarks found in the literature.
\vspace{-0.5cm}
\appendix
\subsection{Useful Corollary and Lemma} 
\begin{corollary}\label{lemmaMatrixRIS}
For a positive semi-definite matrix $\mathbf{D} \in \mathbb{C}^{N\times N}$ and $\widetilde{\mathbf{H}}_m = \mathbf{R}_{m,\mathrm{AP}}^{1/2}\mathbf{X}_m\mathbf{R}_{m,\mathrm{RIS}}^{1/2}$, it holds that 
      $ \mathbb{E}\{\widetilde{\mathbf{H}}_m\mathbf{D}\widetilde{\mathbf{H}}_m^H \} = \Tr\tron{\mathbf{R}_{m,\mathrm{RIS}}\mathbf{D}}\mathbf{R}_{m,\mathrm{AP}}.$
\begin{proof}
The left-hand side of \eqref{lemmaMatrixRIS} can be written as follows 
\begin{equation}
\begin{split}
&\mathbb{E}\{\widetilde{\mathbf{H}}_m\mathbf{D}\widetilde{\mathbf{H}}_m^H \} = \mathbb{E}\{\mathbf{R}_{m,\mathrm{AP}}^{1/2}\mathbf{X}_m\mathbf{R}_{m,\mathrm{RIS}}^{1/2}\mathbf{D}\mathbf{R}_{m,\mathrm{RIS}}^{1/2}\mathbf{X}_m{\mathbf{R}_{m,\mathrm{AP}}^{1/2}} \} \stackrel{(a)}{=} \\
& 
\mathbf{R}_{m,\mathrm{AP}}^{1/2} \Tr(\mathbf{R}_{m,\mathrm{RIS}}^{1/2}\mathbf{D}{\mathbf{R}_{m,\mathrm{RIS}}^{1/2}}) \mathbf{R}_{m,\mathrm{AP}}^{1/2} \stackrel{(b)}{=}  \Tr\tron{\mathbf{R}_{m,\mathrm{RIS}}\mathbf{D}}\mathbf{R}_{m,\mathrm{AP}},
\end{split}
\end{equation}
where $(a)$ is obtained by the independence of propagation channels between the AP and the RIS together with lemma \cite[Lemma 2]{le2023double}; $(b)$ is obtained by the property $\Tr\tron{\mathbf{XY}} = \Tr\tron{\mathbf{YX}}$ for size-matched matrices $\mathbf{X}$ and $\mathbf{Y}$. 
\end{proof} 
\end{corollary}
\begin{lemma}\label{secondmoment}
 The first moment of the aggregated channel $\mathbf{u}_{mk}$ is $\mathbb{E}\nhon{{\mathbf{u}}_{mk}} = \overline{\mathbf{H}}_m\mathbf{\Phi}\overline{\mathbf{z}}_k$. Additionally, the covariance matrix of the NLoS component $\widetilde{\mathbf{u}}_{mk}$ is
\begin{equation}
\mathbb{E}\{\tilde{\mathbf{u}}_{mk}\tilde{\mathbf{u}}_{mk}^H\}=\mathbf{\Delta}_{mk} +  \alpha_{mk}\mathbf{R}_{m,\mathrm{AP}}. 
\end{equation}
Furthermore, the second moment of  $\mathbf{u}_{mk}$ is computed as
    \begin{equation}
\mathbb{E}\nhon{\mathbf{u}_{mk}\mathbf{u}_{mk}^H} =\overline{\mathbf{H}}_m\mathbf{\Phi}\overline{\mathbf{z}}_k\overline{\mathbf{z}}_k^H\mathbf{\Phi}^H \overline{\mathbf{H}}_m^H + \mathbf{\Delta}_{mk} +  \alpha_{mk}\mathbf{R}_{m,\mathrm{AP}}.
\end{equation}
 The correlation between the NLoS components of the two aggregated channels $\mathbf{u}_{mk}$ and $\mathbf{u}_{mk'}$, $\forall k' \ne k$, is
    \begin{equation}
\mathbb{E}\{\widetilde{\mathbf{u}}_{mk}\widetilde{\mathbf{u}}_{mk'}^H\} = \Tr(\mathbf{R}_{m,\mathrm{RIS}}\mathbf{\Phi}\overline{\mathbf{z}}_{k}\overline{\mathbf{z}}_{k'}^H\mathbf{\Phi}^H)\mathbf{R}_{m,\mathrm{AP}}.
    \end{equation}
\end{lemma}
\begin{proof}
The first moment of the aggregated channel $\mathbf{u}_{mk}$ is $\mathbb{E}\{\mathbf{u}_{mk}\} = \mathbb{E}\{\mathbf{g}_{mk} + (\overline{\mathbf{H}}_m + \widetilde{\mathbf{H}}_m)\mathbf{\Phi}\tron{\overline{\mathbf{z}}_{k} + \widetilde{\mathbf{z}}_{k}}\} = \mathbb{E}\{
        \mathbf{g}_{mk} + \overline{\mathbf{H}}_m\mathbf{\Phi} \overline{\mathbf{z}}_{k} + \widetilde{\mathbf{H}}_m\mathbf{\Phi}\overline{\mathbf{z}}_{k} + \overline{\mathbf{H}}_m\mathbf{\Phi}\widetilde{\mathbf{z}}_{k} + \widetilde{\mathbf{H}}_m\mathbf{\Phi}\widetilde{\mathbf{z}}_{k}\}=\overline{\mathbf{H}}_m\mathbf{\Phi} \overline{\mathbf{z}}_{k} $. Moreover,
the second moment of the NLoS component of $\mathbf{u}_{mk}$ is
\begin{multline}
 \mathbb{E}\{\tilde{\mathbf{u}}_{mk}\tilde{\mathbf{u}}_{mk}^H \}
= 
 \mathbf{G}_{mk} + \Tr(\mathbf{R}_{m,\mathrm{RIS}}\mathbf{\Phi}\overline{\mathbf{z}}_{k}\overline{\mathbf{z}}_{k}^H\mathbf{\Phi}^H)\mathbf{R}_{m,\mathrm{AP}} +\overline{\mathbf{H}}_m \times\\
 \mathbf{\Phi}
 \mathbf{R}_{k}\mathbf{\Phi}^H\overline{\mathbf{H}}_m^H   + \Tr(\mathbf{R}_{m,\mathrm{RIS}}\mathbf{\Phi}\mathbf{R}_{k}\mathbf{\Phi}^H)\mathbf{R}_{m,\mathrm{AP}} = \mathbf{\Delta}_{mk} + \alpha_{mk}\mathbf{R}_{m, \mathrm{AP}},
\end{multline}
by virtue of Corollary \ref{lemmaMatrixRIS}. The second moment of the channel $\mathbf{u}_{mk}$ is $\mathbb{E}\{\mathbf{u}_{mk}\mathbf{u}_{mk}^H \} = \mathbb{E}\{\tron{\overline{\mathbf{u}}_{mk} + \widetilde{\mathbf{u}}_{mk}}\tron{\overline{\mathbf{u}}_{mk} + \widetilde{\mathbf{u}}_{mk}}^H \}
    \stackrel{(a)}{=} \overline{\mathbf{u}}_{mk}\overline{\mathbf{u}}_{mk}^H + \mathbb{E}\nhon{\widetilde{\mathbf{u}}_{mk}\widetilde{\mathbf{u}}_{mk}^H} = \overline{\mathbf{H}}_m\mathbf{\Phi}\overline{\mathbf{z}}_k\overline{\mathbf{z}}_k^H\mathbf{\Phi}^H \overline{\mathbf{H}}_m^H + \mathbf{\Delta}_{mk} +  \alpha_{mk}\mathbf{R}_{m,\mathrm{AP}}$,
where $(a)$ is obtained because of the independence between the LoS and NLoS components of the channel $\mathbf{u}_{mk}$. Moreover, the correlation between the two NLoS parts of the aggregated channels $\mathbf{u}_{mk}$ and $\mathbf{u}_{mk'}$, $k' \ne k$, is as 
$\mathbb{E}\{\widetilde{\mathbf{u}}_{mk}\widetilde{\mathbf{u}}_{mk'}^H\} 
=\Tr(\mathbf{R}_{m,\mathrm{RIS}}\mathbf{\Phi}\overline{\mathbf{z}}_{k}\overline{\mathbf{z}}_{k'}^H\mathbf{\Phi}^H)\mathbf{R}_{m,\mathrm{AP}},
$
by virtue of Corollary \ref{lemmaMatrixRIS}.
    \end{proof}
   \vspace{-0.2cm}
\subsection{Proof of Lemma~\ref{Lemma:ChannelEst}}
We first define the term $\widetilde{\mathbf{y}}_{pmk}, \forall m,k,$ in \eqref{estimatedchannel} as
\begin{equation}
\begin{split}
&\widetilde{\mathbf{y}}_{pmk} = \mathbf{y}_{pmk} - \mathbb{E}\{\mathbf{y}_{pmk} \} =  \displaystyle\sqrt{p\tau_p}\sum\nolimits_{k' \in \mathcal{P}_k} \mathbf{u}_{mk'} + \mathbf{w}_{pmk} \\
& - \tron{\sqrt{p\tau_p}\sum\nolimits_{k' \in \mathcal{P}_k} \overline{\mathbf{u}}_{mk}} = \sqrt{p\tau_p}\sum\nolimits_{k' \in \mathcal{P}_k} \widetilde{\mathbf{u}}_{mk'} + \mathbf{w}_{pmk}.
\end{split}
\end{equation}
By exploiting the above expression of $\widetilde{\mathbf{y}}_{pmk}$, the expectation $\mathbf{\Gamma}_{mk} = \mathbb{E}\{\widetilde{\mathbf{u}}_{mk}\widetilde{\mathbf{y}}_{pmk}^H\}$ can be formulated as
\begin{equation}
\begin{split}
&\mathbf{\mathbf{\Gamma}}_{mk} =
\mathbb{E}\nhon{\widetilde{\mathbf{u}}_{mk}\tron{\sqrt{p\tau_p}\sum\nolimits_{k' \in \mathcal{P}_k} \widetilde{\mathbf{u}}_{mk'}^H + \mathbf{w}_{pmk}}^H} = \sqrt{p\tau_p} \times \\
& 
 \sum\nolimits_{k' \in \mathcal{P}_k \setminus \nhon{k}} \mathbb{E}\nhon{\widetilde{\mathbf{u}}_{mk}\widetilde{\mathbf{u}}_{mk'}^H} + \sqrt{p\tau_p} \mathbb{E}\nhon{\widetilde{\mathbf{u}}_{mk}\widetilde{\mathbf{u}}_{mk}^H}  \stackrel{(a)}{=} \sqrt{p\tau_p} \times \\
&   \sum\nolimits_{k' \in \mathcal{P}_k} \Tr\tron{\mathbf{R}_{m,\mathrm{RIS}}\mathbf{\Phi}\overline{\mathbf{z}}_{k}\overline{\mathbf{z}}_{k'}^H\mathbf{\Phi}^H}\mathbf{R}_{m,\mathrm{AP}} +\sqrt{p\tau_p}\mathbf{\Delta}_k \stackrel{(b)}{=} \sqrt{p\tau_p} \\
& \times  \Tr\tron{\mathbf{R}_{m,\mathrm{RIS}}\mathbf{\Phi}\overline{\mathbf{z}}_{k}\tron{\sum\nolimits_{k' \in \mathcal{P}_k} \overline{\mathbf{z}}_{k'}^H}\mathbf{\Phi}^H}\mathbf{R}_{m,\mathrm{AP}}+\sqrt{p\tau_p}\mathbf{\Delta}_{mk},
\end{split}
\end{equation}
where $(a)$ follows Lemma~\ref{secondmoment}; and $(b)$ is obtained by exploiting the linearity of the trace of matrices. Furthermore,
the covariance matrix $\mathbf{\Psi}_{mk} = \mathbb{E}\{\widetilde{\mathbf{y}}_{pmk}\widetilde{\mathbf{y}}_{pmk}^H\}$ can be formulated as
\begin{equation}
\begin{split}
&\mathbf{\mathbf{\Psi}}_{mk} =\\&
\mathbb{E}\nhon{\tron{\sqrt{p\tau_p}\sum_{k' \in \mathcal{P}_k} \widetilde{\mathbf{u}}_{mk'} + \mathbf{w}_{pmk}}\tron{\sqrt{p\tau_p}\sum_{k' \in \mathcal{P}_k} \widetilde{\mathbf{u}}_{mk'}^H + \mathbf{w}_{pmk}}^H} \\ 
& \stackrel{(a)}{=} \sum_{\substack{k',k'' \in \mathcal{P}_k \\ k' \ne k''}} \Tr\tron{\mathbf{R}_{m,\mathrm{RIS}}\mathbf{\Phi}\overline{\mathbf{z}}_{k'}\overline{\mathbf{z}}_{k''}^H\mathbf{\Phi}^H}\mathbf{R}_{m,\mathrm{AP}} + \sum_{k' \in \mathcal{P}_k} p\tau_p\mathbf{\Delta}_{mk} + p\tau_p  \\
& \times \sum_{k' \in \mathcal{P}_k} \Tr\tron{\mathbf{R}_{m,\mathrm{RIS}}\mathbf{\Phi}\overline{\mathbf{z}}_{k'}\overline{\mathbf{z}}_{k'}^H\mathbf{\Phi}^H}\mathbf{R}_{m,\mathrm{AP}} +  \mathbf{I}_M  = p\tau_p\sum_{k' \in \mathcal{P}_k} \mathbf{\Delta}_{mk'} \\
& + \mathbf{I}_M  + p\tau_p \Tr\tron{\mathbf{R}_{m,\mathrm{RIS}}\mathbf{\Phi}\tron{\sum_{k' \in \mathcal{P}_k} \overline{\mathbf{z}}_{k'}}\tron{\sum_{k' \in \mathcal{P}_k} \overline{\mathbf{z}}_{k'}^H}\mathbf{\Phi}^H}\mathbf{R}_{m,\mathrm{AP}} ,
\end{split}
\end{equation}
where $(a)$ follows by Lemma~\ref{secondmoment} and the channel estimate obtained as in the lemma.

\bibliographystyle{IEEEtran}
\bibliography{IEEEabrv,refs}
\end{document}